\documentclass[preprint,12pt]{elsarticle}

\newtheorem{theorem}{Theorem}
\newtheorem{corollary}{Corollary}

\newenvironment{proof}{\textbf{Proof}}{\eop}

\usepackage{amsmath,amssymb,graphicx,latexsym,soul,color,hyperref,setspace,mathrsfs,framed,pdflscape,xcolor}\hypersetup{pdftex,colorlinks=true,linkcolor=blue,citecolor=blue,filecolor=blue,urlcolor=blue}
\def\eop{\hfill{$\Box$}\medskip}

\usepackage{xcolor}

\usepackage{lineno}
  \usepackage{setspace}
\begin{document}

\begin{frontmatter}

\title{Demystification of Entangled Mass Action Law}

\author[NNU,ICM]{A. N. Kirdin\corref{cor1}}
\ead{kirdinalexandergp@gmail.com}
\author[NNU]{S. V. Stasenko}
\ead{stasenko@neuro.nnov.ru}

\address[NNU]{Lobachevsky University, Nizhni Novgorod, Russia}
\address[ICM]{Institute for Computational Modelling, Russian Academy of Sciences, Siberian Branch, Krasnoyarsk, Russia}
 \cortext[cor1]{Corresponding author}

\date{}

\begin{abstract}
Recently, Gorban (2021) analysed some kinetic paradoxes of the transition state theory and proposed its revision that gave the ``entangled mass action law'', in which new reactions were generated as an addition to the reaction mechanism under consideration.
These paradoxes arose due to the assumption of quasiequilibrium between reactants and transition states.
 
In this paper, we provided a brief introduction to this theory, demonstrating how the entangled mass action law equations can be derived in the framework of the standard quasi steady state approximation in combination with the quasiequilibrium generalized mass action law for an auxiliary reaction network including reactants and intermediates.  We also proved the basic physical property (positivity) for these new  equations, which was not obvious in the original approach. 
\end{abstract}
\end{frontmatter}

\section{Introduction}

Mass Action Law (MAL) exists in two basic forms: 
\begin{enumerate}
\item Equilibrium (or static) MAL that describes chemical equilibria by the systems of algebraic equations; 
\item Dynamic (or kinetic) MAL invented for description of chemical dynamics.
\end{enumerate}
Both versions were proposed by Guldberg and Waage in a series of papers in 1864--1879. 

In physical chemistry, the static MAL was developed to the most general thermodynamic form for perfect and non-perfect systems by Gibbs \cite{Gibbs1875}. The physical justification of the dynamic MAL was provided in 1935 simultaneously by Eyring, Polanyi and Evans. They introduced transition states (activated complexes) as universal intermediates in chemical reactions (we refer to the analytic review \cite{Laidler1983} for the basic notions and further references). 

The static MAL was used in  Transitions State Theory (TST) to describe the quasiequilibrium between the reactants and the transition state. Recently it was demonstrated that the assumption about quasiequilibrium between reactants and transition states leads to some paradoxes in modelling of multistage or reversible reactions \cite{Perez-Benito2017,Gorban2021}. A new kinetic framework for the TST was proposed. 

Surprisingly, these new models led to the same (generalised) MAL expressions but with the entanglement effect: the MAL rate of some elementary reactions are also included in the reaction rates of other reactions. These results created a  mystery: if we tried to abandon MAL quasiequilibrium assumption then we came to the same MAL with the shuffled reaction rates. This puzzle must be solved. Some technical questions also remain open. For example, positivity of the quasi steady state concentrations of the intermediates should be proven. Explicit formulas for entropy production are also very desirable.

{In this work, we obtained the following results:
\begin{enumerate}
\item We demonstrated how the new entangled MAL equations may be derived in the framework of the standard quasi steady state assumption combined with the quasiequilibrium generalised mass action law. 
\item We proved a basic physical property (positivity) of the new entangled MAL equations, which was not obvious in the original work \cite{Gorban2021}
\end{enumerate}
}
 
In Sec.~\ref{Sec:Paradoxes} we briefly describe the kinetic paradoxes in the TST. In Sec.~\ref{Sec:Fundam} the fundamentals of entangled MAL are presented with the Positivity Theorem and a  simple (perhaps, the simplest) example. Sec.~\ref{Sec:Demyst} demystifies the entangled MAL theory and represents a complex reaction as a network of the first order transitions between intermediates and generalised MAL transitions between the complexes of the reactants and the corresponding intermediates.

\section{Quasiequilibrium paradox in transition state theory \label{Sec:Paradoxes}}

The most prominent approach to justifying Mass Action Law (MAL)  was provided by Eyring, Polanyi and Evans \cite{Laidler1983}. They introduced transition states (activated complexes) as universal intermediates in chemical reactions. The basic textbook scheme is (\ref{actcompl}\begin{equation}\label{actcompl}
A+B \rightleftharpoons [A-B]  \to \mbox{ Products }.
\end{equation}
Here we use the notation $[A-B]$ for the transition state or activated complex.
The key assumption was that the activated complexes are in quasi-equilibrium with the reactants. 
Therefore, the quasiequilibrium concentrations of the activated complex can be estimated using thermodynamics, and the overall reaction rate is the product of this concentration and the reaction rate constant for the $[A-B] \to \mbox{ Products }$ transition.

An additional assumption is the low concentration of the activated complex compared to the concentrations of the reactants. Without this hypothesis, the MAL formulas cannot be produced \cite{Gorban2011}. 

Thus, the problem of estimating the reaction rate was divided into two tasks:
\begin{itemize}
\item {\it Thermodynamic}  equilibration   $A+B \rightleftharpoons [A-B]$;
\item {\it Dynamic}  evaluation of transition rate $[A-B]  \to \mbox{ Products }$.  
\end{itemize}

This nice picture hides several problems. First of all, the reaction can be reversible.  Consider, just for simplicity, products $C+D$ (\ref{actcompl2}): 
 \begin{equation}\label{actcompl2}
A+B \rightleftharpoons [A-B]  \rightleftharpoons C+D.
\end{equation}
Microreversibility and detailed balance require that the reverse reaction in (\ref{actcompl2}) follows the same route as the initial reaction. According to the quasiequilibrium assumption, the reaction $[A-B]  \rightleftharpoons C+D$ should be also in quasiequilirium. Simple algebra demonstrates that these two quasiequilibrium assumptions imply complete equilibrium and reaction vanishes \cite{Perez-Benito2017, Gorban2021}. Two solution to this paradox were proposed:
\begin{enumerate}
\item Consider the model with two intermediates and transition between them:
$$A+B \rightleftharpoons [A-B] \rightleftharpoons   [C-D] \rightleftharpoons C+D.$$
Asymptotic assumptions about two quasiequilibria, $A+B \rightleftharpoons [A-B]$ and $[C-D] \rightleftharpoons C+D$, and smallness of the $[A-B]$ and $[C-D]$ concentrations lead to classical MAL in very wide conditions \cite{Gorban2011, Gorban2015}.
\item Abandon the quasiequilibrium hypothesis but keep the assumption that the concentration of the active complex is much smaller than that of  the reactants and prove this assumption when possible \cite{Perez-Benito2017}. This assumption violates the polynomial MAL and leads to more complex rational reaction rate dependencies.
\end{enumerate}

Both approaches have a long history. Combination of quasiequilibrium and small concentration assumptions for intermediate compounds was used by Michaelis and Menten in 2013 \cite{Michaelis1913}. Stueckelberg in 1952 \cite{Stueckelberg1952} used the same two assumptions for analysis of the Boltzmann equation beyond microreversibility and proved general semidetailed balance that is known now also as cyclic balance or complex balance. The quasiequilibrium condition in enzyme kinetics was abolished by \cite{Briggs1925}. (For more modern analysis we refer to the work by \cite{SegelSlemrod1989}.) They assumed only the smallness of intermediate concentrations and obtained a nonpolynomial reaction rate, which was called the Michelis--Menten kinetics. The same formula was proposed by \cite{Perez-Benito2017} for general transition state kinetics. 

This approach may give correct answers but has some logical issues: We aim to justify MAL for general (non-linear) kinetics. The transition state theory uses thermodynamic definition of quasiequilibrium (the static MAL) and simple first order Markov kinetics for transition of activated complex. The result is the MAL kinetics for nonlinear reactions of arbitrary complexity. But if we would like to apply the Briggs--Haldane approach then we must assume dynamic MAL for all elementary transitions from scratch, before justifying.

{
The assumption about small concentrations of the intermediates was used explicitly in enzyme kinetics \cite{Michaelis1913, Briggs1925}, in gas kinetics \cite{Stueckelberg1952}, and in kinetics of heterogeneous catalytic reactions \cite{YBGE1991}.
In TST, it is used usually implicitly and, therefore, needs further clarification. Let us take the basic example from the popular textbook \cite{Atkins2006} (Section ``Transition State Theory"): 
$A+B \rightleftharpoons  C \to P,$
where $A$ and $B$ are the reactants, $C$ is the activated complex and $P$ is the product.
The fast quasiequilibrium assumption gives $[C]=K[A][B]$. The TST produces an estimate of the reaction rate constant $k$ for the transition $C\to P$. After that, we have to exclude $[C]$ from the material balance equations using the quasiequilibrium assumption.
 
The material balance gives
\begin{equation}\label{EqTNTkin}
\frac{d[P]}{dt}=-\frac{dM}{dt}=kK[A][B],
\end{equation}
where $M=[A]+[B]+[C]$
Notice that $\Delta=[B]-[A]$ does not change in the reaction.
The quasiequilibrium assumption provides the {\it quadratic} equation for $[A]$:
\begin{equation}\label{Eq:quadraticQE}
\begin{split}
&2[A]+\Delta+K[A]([A]+\Delta)=M; \\
&[A]=\frac{\sqrt{4+K^2\Delta^2+4KM}-2-K\Delta}{2K}; \\
&[B]=\frac{\sqrt{4+K^2\Delta^2+4KM}-2+K\Delta}{2K}; 
\end{split}
\end{equation}
(the solutions with positive concentrations are selected).

Even for this simple example, the correct kinetic equation with the quasiequilibrium assumption but without smallness of the intermediate concentration $[C]$ differs qualitatively  from the simple kinetics for the elementary reaction $A+B  \to P$. This is not a miracle, because the intermediate with non-small concentration is an additional reservoir for the substances that modifies the reaction rate. 

If we assume that $[C]$ is small then the solution of the quadratic equations can be simplified:
$$[A]=\frac{M-\Delta}{2} +o([C]), \; [B]=\frac{M+\Delta}{2} +o([C]),$$
and the reaction rate has the same form as for the elementary reaction $A+B \to P$: $v=kK(M^2-\Delta^2)/4$.

The assumption about smallness of concentrations can appear in various forms: short lifetime of intermediates, small equilibrium constant, etc. Most of these assumptions are essentially the same but accurate reconciliation between them is needed for the proper slow/fast separation. For example, if $[C]$ is small then the reaction rate constant $k$ should be large in order to make the reaction rate non-negligible. On another hand, if we assume that the life time of the  intermediates is not small (and, hence, $k$ is not large) then their concentrations should not be small in order to have a non-zero asymptotic reaction rate. 
}

The approach with several intermediate states and first order kinetic transitions between them \cite{Stueckelberg1952,Gorban2011,Gorban2015} was extended and applied (independently) to the transition state theory \cite{DiGezu2017} with introduction of many intermediates and  Markov kinetics of their transformations.

Let us start formally from the scheme with two intermediates, $B_{\rho}^+$, $B_{\rho}^-$ for each elementary reaction. A complex reaction is represented by the system of stoichiometric equations:
\begin{equation}\label{Eq:StueckStoich}
\sum_i \alpha_{\rho i}A_i\rightleftharpoons B_{\rho}^+ \to B_{\rho}^- \rightleftharpoons \sum_i \beta_{\rho i}A_i\ ,
\end{equation}
where $A_i$ are components (substances), $\rho$ is the number of the elementary reaction,  $B_{\rho}^{\pm}$ are intermediate compounds, and $\alpha_{\rho i}, \beta_{\rho i} \geq 0$ are stoichiometric coefficients, usually non-negative integers.

Two asymptotic assumptions about (1) smallness of the $B_{\rho}^{\pm}$ concentrations (comparing to the concentrations of $A_i$) and (2) fast quasiequilibria of the reversible reactions $\sum_i \alpha_{\rho i}A_i\rightleftharpoons B_{\rho}^+$ and  ${B_{\rho}}^- \rightleftharpoons \sum_i \beta_{\rho i}A_i$ unambiguously entail the consequence: in this asymptotic the intermediates can be excluded and the ``brutto'' reaction mechanism (\ref{Eq:BruttoStoich})
\begin{equation}\label{Eq:BruttoStoich}
\sum_i \alpha_{\rho i}A_i \to \sum_i \beta_{\rho i}A_i
\end{equation}
satisfies the Generalised Mass Action Law (GMAL) with reaction rate (\ref{Eq:GMAL})
\begin{equation}\label{Eq:GMAL}
r_{\rho}=\varphi_{\rho}\exp\left(\sum_i \alpha_{\rho i}\frac{\mu_i}{RT}\right)\ ,
\end{equation}
where $\varphi_{\rho} \geq 0$ are non-negative variables (``kinetic factors''), $\mu_i$ are the chemical potentials of $A_i$, $R$ is the universal gas constant, $T$ is the temperature, and $\exp(\sum_i \alpha_{\rho i}{\mu_i}/{RT})$ is the purely thermodynamic ``Boltzmann 
factor''.

At this stage, we followed Stueckelberg \cite{Stueckelberg1952} and did not assume a detailed balance or entropy growth in the  reaction network. Nevertheless,  in this limit, the kinetic factors, $\varphi_{\rho}$ (\ref{Eq:GMAL}) always satisfy the semidetailed balance condition (known also as the cyclic or complex balance)
\begin{equation}\label{complexbalance}
\sum_{\rho, \,\alpha_{\rho}=y} \varphi_{\rho}\equiv
\sum_{\rho, \,\beta_{\rho}=y} \varphi_{\rho}
\end{equation}
for any vector $y$ from the set of all vectors
$\{\alpha_{\rho}, \beta_{\rho}\}$.
This statement was proven for the Boltzmann collision by Stueckelberg \cite{Stueckelberg1952}. Later it was generalised  for the general chemical reaction mechanisms  \cite{Gorban2011}, and for the general nonlinear Markov processes \cite{Gorban2015}. It is worth to mention that for the non-trivial limit the reaction rates of the compounds should be properly scaled (to $\infty$) when their concentrations tend to 0 to keep the proper order of their products \cite{Gorban2021}.

{
Already a simple example demonstrates that the identities (\ref{complexbalance}) are
weaker than the detailed balance conditions. Consider the reaction mechanism: (1) $2A_1 \to 2 A_2$, (2) $2A_2 \to A_1 +A_3$, (3) $A_1+ A_3 \to 2 A_3$, (4) $2A_3\to 2A_1$, and (5) $A_1+A_3 \to 2A_1$. In particular, this scheme resembles some mechanisms of the surface reactions in heterogeneous catalysis \cite{Marin2019}. To formulate the semidetailed balance conditions, we should prepare the list of stoichiometric vectors $\alpha_{\rho}$ and 
$\beta_{\rho}$
\begin{equation*}
\begin{split}
&\alpha_1=\beta_4=\beta_5=\left(\begin{array}{c}2 \\ 0 \\ 0 \end{array}\right); \;
\alpha_2=\beta_1= \left(\begin{array}{c}0 \\ 2 \\ 0 \end{array}\right); \\ 
&\alpha_3=\alpha_5=\beta_2=\left(\begin{array}{c}0 \\ 1 \\1 \end{array}\right); \;
\alpha_4=\beta_3=\left(\begin{array}{c}0 \\ 0 \\2 \end{array}\right). 
\end{split}
\end{equation*}
There are four different vectors in this list, therefore, there are four complex balance identities (\ref{complexbalance}): 
\begin{enumerate}
\item $\varphi_1=\varphi_4+\varphi_5$,
\item $\varphi_2=\varphi_1$,
\item $\varphi_2=\varphi_3+\varphi_5$,
\item $\varphi_3=\varphi_4$.
\end{enumerate} 
Only three identities of them are independent.
After obvious simplifications, we obtain a two-dimensional cone of the vectors $(\varphi_i)$ of complexly balanced kinetic factors  in the five-dimensional positive orthant. This cone is parametrized by $\varphi_3$ and $\varphi_5$: $\varphi_i>0$,   $\varphi_1=\varphi_2=\varphi_3+\varphi_5$, $\varphi_4=\varphi_3$. All the reactions are irreversible, hence, the detailed balance conditions do not hold. 
}

Nevertheless, systems with microreversibility, detailed balance, and positive equilibria form the most common and well studied class of chemical kinetics.  For them, the direct and reverse reactions can be coupled in one stoichiometric equation (\ref{Eq:StueckStoichDB})
\begin{equation}\label{Eq:StueckStoichDB}
\sum_i \alpha_{\rho i}A_i\rightleftharpoons B_{\rho}^+ \rightleftharpoons B_{\rho}^- \rightleftharpoons \sum_i \beta_{\rho i}A_i\ ,
\end{equation}
The asymptotic limit has also the reversible form (\ref{Eq:BruttoStoichDB}) 
\begin{equation}\label{Eq:BruttoStoichDB}
\sum_i \alpha_{\rho i}A_i \rightleftharpoons \sum_i \beta_{\rho i}A_i\ ,
\end{equation}
and the conditions (\ref{complexbalance}) transform into the beautiful {\it detailed balance} conditions:
\begin{equation*}
\varphi_{\rho}^+=\varphi_{\rho}^-,
\end{equation*}
here, we can omit the $\pm$ superscripts: $\varphi_{\rho}^+=\varphi_{\rho}^-=\varphi$.  
For reversible systems with detailed balance, it is convenient to factorise the reaction rate of reversible reaction into non-negative kinetic and thermodynamic (Boltzmann's) factors:
\begin{equation}
r_{\rho}=r_{\rho}^+-r_{\rho}^-=\varphi_{\rho}\left(\exp\left(\sum_i \alpha_{\rho i}\frac{\mu_i}{RT}\right) - \exp\left(\sum_i \beta_{\rho i}\frac{\mu_i}{RT}\right)\right).
\end{equation} 

GMAL was invented in order to meet the thermodynamic restrictions on kinetics \cite{Grmela2010, Pavelka2018}. The asymptotic limit  \cite{Stueckelberg1952, Gorban2011, Gorban2015} demonstrated even more: GMAL can be produced from the classical equilibrium thermodynamics and two asymptotic assumptions: fast quasiequilibria and small concentrations of intermediates. Without any additional dynamic assumption like microreversibility, this asymptotic analysis gives also the semidetailed/cyclic/complex balance conditions (\ref{complexbalance}) that guarantees the thermodynamic properties of the models. For the systems with microreversibility it transforms into the detailed balance conditions. 

Table \ref{Tab:GMALtoMAL} presents a vocabulary to translate general expressions of GMAL to the particular case of the textbook MAL \cite{Marin2019}.

\begin{table}
\caption{GMAL to MAL translation. \label{Tab:GMALtoMAL} Here, $c_i$ is the concentration of $A_i$, $c^*=(c_i^*)$ is the positive vector of standard equilibrium where all $\mu_i=0$.} 
\begin{center}
\begin{tabular}{ l l } 
 \hline
GMAL expression & MAL expression \\ 
\hline
 \\
 Free energy density $g(c,T)$ & $RT\sum_i c_i \left(\ln\frac{c_i}{c_i^*} - 1\right)$ \vspace{0.15 cm}  \\ 
Chemical potential $\mu_i = \left(\frac{\partial g}{\partial c_i}\right)_{T,V}$ & $RT\ln\frac{c_i}{c_i^*}$ \vspace{0.15 cm}  \\
Boltzmann factor $\exp\left(\sum_i \alpha_{\rho i}\frac{\mu_i}{RT}\right)$ & $\prod_i \left(\frac{c_i}{c_i^*}\right)^{\alpha_{\rho i}}$ \vspace{0.15 cm}  \\ 
Kinetic factor $\varphi_{\rho}$ & $\varphi_{\rho}= k_{\rho}\prod_i {c_i^*}^{\alpha_{\rho i}}$ \vspace{0.15 cm}   \\  
React. rate $r_{\rho}=\varphi_{\rho}\exp\left(\sum_i \alpha_{\rho i}\frac{\mu_i}{RT}\right)$ & 
$r_{\rho}=k_{\rho}\prod_i {c_i}^{\alpha_{\rho i}}	$ \vspace{0.15 cm} \\  
\hline
\end{tabular}
\end{center}
\end{table}

Existence of the standard {\it positive} equilibrium is, surprisingly, a non-trivial question. It is connected with the perfect asymptotic expressions for the free energy of small admixtures and the logarithmic singularities of the chemical potentials at the border of the positive orthant of concentrations. The limit of kinetic systems when some $c_i^*\to 0$ depends on the rates of convergence to zero for different $c_i^*$ and their ratio. These questions need to be discussed separately are partially answered \cite{GorbanYabl2011} for systems with detailed balance. 

Table~\ref{Tab:GMALtoMAL} demonstrates the physical sense of the kinetic factors $\varphi_{\rho}$ for MAL. They are just the reaction rates at the standard equilibria.

The equations of chemical kinetics (for $T,V=const$) have the form (\ref{Eq:KinUr})
\begin{equation}\label{Eq:KinUr}
\frac{d c_i}{dt}=\sum_{\rho} r_{\rho} \gamma_{\rho i},
\end{equation}
where  $\gamma_{\rho i}=\beta_{\rho i }-\alpha_{\rho i}$.

If the system is produced from (\ref{Eq:StueckStoich}) in the limit of fast quasiequilibria and small intermediates then the reaction rates satisfy the semidetailed/cyclic/complex balance conditions (\ref{complexbalance}) and the free energy changes monotonically in time: in accordance to (\ref{Eq:KinUr}), $dg/dt \leq 0$  \cite{Gorban2011}. Especially simple form $dg/dt$ takes for the systems with detailed balance \cite{Gorban2013}: 
$$\frac{dg}{dt}=-RT\sum_{\rho}(r_{\rho}^++r_{\rho}^-)\mathbb{A}_{\rho}\tanh\frac{\mathbb{A}_{\rho}}{2}\leq 0,$$
where $\mathbb{A}_{\rho}=\sum_i \gamma_{\rho i} \mu_i/RT$ is {\em affinity}, the sum of reaction rates of direct and reverse elementary reactions, $r_{\rho}^++r_{\rho}^-$, is non-negative as well as the product $\mathbb{A}_{\rho}\tanh(\mathbb{A}_{\rho}/{2})$.   

For the non-isochoric conditions, it is more convenient to write the equations for the amounts of $A_i$, the extensive variables  $N_i=Vc_i$: ${d N_i}{dt}=V\sum_{\rho} r_{\rho} \gamma_{\rho i}$. Their thermodynamic properties are also well-studied \cite{Gorban2011, Hangos2010}, 

The MAL equations without any conditions on the reaction rate constant can approximate any dynamics on the reaction polyhedron (that is the intersection of the positive cone with the linear manifold with given values of the linear conservation laws). The GMAL equations constructed in concordance with thermodynamics have additional restriction in the form of semidetailed/cyclic/complex balance (\ref{complexbalance}) (for Markov microscopic kinetics) or detailed   balance (for reversible Markov microscopic kinetics). 

There may be additional difficulty: a system of kinetic equations can have non-unique  MAL representation. For more detailed analysis we refer to the detailed work \cite{Hangos2011}, 
where numerical procedures are proposed and tested for finding complex balanced or detailed balanced realizations of mass action type chemical reaction
networks.

\section{The fundamentals of entangled MAL \label{Sec:Fundam}}

It seems that the quasiequilibrium paradox for the transition state theory was resolved many  years ago \cite{Stueckelberg1952, Gorban2011, Gorban2015}. But another problem was noticed very recently \cite{Gorban2021}: the assumption about small concentration of intermediates implies that the reaction rate constant of the active complex reverse decomposition should be much larger than the rate constant of its transition towards the product.

For illustration of this statement, consider a very simple example of MAL system:
$$A_1+A_2  \rightleftharpoons  B \to \ldots$$
with reaction rate constants $k_1$ (for $A_1+A_2 \to B$), $k_{-1}$ (for $B \to A_1+A_2$), and $\kappa$ (for $B\to \ldots$).

Note, that neither MAL nor GMAL assumption for the reaction rates is necessary in the transition state theory. Moreover, this theory should produce the `proper' reaction rates from the microscopic dynamics of activated complexes and thermodynamic description of quasiequilibria without dynamic assumptions. Nevertheless, for this simple instructive example we select the system that obeys known dynamic MAL kinetic law at each step.

{
The kinetic equations are
\begin{equation*}
\begin{split}
& \frac{d c_1}{dt}=\frac{d c_2}{dt}=-k_1 c_1c_2+k_{-1}c_B,\\
& \frac{dc_B}{dt}=k_1  c_1c_2-(k_{-1}+\kappa)c_B.
\end{split}
\end{equation*}
}

Scaling for `fast equilibrium' assumes introduction of a small parameter $\varepsilon$ in the rate constants $k_{\pm 1}$: $k_{\pm 1} \leftarrow \frac{1}{\varepsilon} k_{\pm 1}$. Smallness of the concentration of $B$ requires that the equilibrium  $A_1+A_2  \rightleftharpoons  B$ should be strongly shifted to the left. One more small parameter is needed: $k_{-1} \leftarrow \frac{1}{\delta} k_{-1}$. Thus, the reaction rate constants should include two small parameters. These constants are: $\frac{1}{\varepsilon} k_{1}$ (for $A_1+A_2 \to B$) and $\frac{1}{\varepsilon \delta} k_{-1}$ for (for $B \to A_1+A_2)$. 

The quasiequilibrium concentration of $B$ is $c_B=\frac{\delta k_1}{k_{-1}}c_1c_2$ and the brutto reaction rate in the quasiequilibrium approximation is $r=\frac{\kappa  \delta k_1}{k_{-1}} c_1 c_2$ (for $A_1+A_2 \to \ldots$). We can see that for the proper asymptotic $\kappa$ should be scaled as $\delta^{-1}$: $\kappa \leftarrow \frac{1}{\delta} \kappa$. 
{ With these small parameters, the kinetic equations are
\begin{equation}\label{small concentrat}
\begin{split}
& \frac{d c_1}{dt}=\frac{d c_2}{dt}=-\frac{1}{\varepsilon}k_1 c_1c_2+\frac{1}{\varepsilon \delta}k_{-1}c_B,\\
& \frac{dc_B}{dt}=\frac{1}{\varepsilon}k_1  c_1c_2-\left(\frac{1}{\varepsilon \delta}k_{-1}+\frac{1}{\delta}\kappa\right)c_B.
\end{split}
\end{equation}
} 
In this case, the reaction rate for the observable brutto reaction does not include small parameters and is $r=\frac{\kappa k_1}{k_{-1}} c_1 c_2$. Thus, in the proper explicit scaling, there are two transitions from the intermediate $B$: (1) the reverse transition to the initial reactants $B\to A_1+A_2$ with the reaction rate constant $\frac{1}{\varepsilon \delta} k_{-1}$ and  (2) the reaction towards the products $B\to \ldots$ with the reaction rate constant $\frac{1}{\delta} \kappa$. Here, $k_{-1}$ and $\kappa$ are the reaction rate constants before scaling, and the positive small parameters $\varepsilon$ and $\delta$ are used for description of the asymptotics. 

We can see, that together with the ``natural'' assumptions about quasiequilibrium between reactants and intermediate compounds and about small concentrations of intermediates we inevitably approach a strong asymmetry in microscopic transitions: the reverse transformation of the intermediates into reactants is much faster than its reaction towards products (Fig.~\ref{Fig:Asymmetry}). Formally, this is not a contradiction, because the overall (`brutto') reaction rate satisfies the MAL with the reaction rate constant that does not depend on the scaling parameters $\varepsilon$ and $\delta$. Nevertheless, from the physical point of view such asymmetry looks unreasonable and needs either a further explanation or a demonstration that in the next approximation MAL persists. Surprisingly, MAL is valid beyond the quasiequilibrium assumption for transition states \cite{Gorban2021}.  

\begin{figure}
\centering{
\includegraphics[width=0.5 \columnwidth]{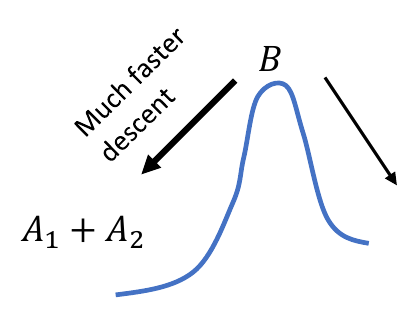}
 \caption{\label{Fig:Asymmetry} {{\it Asymmetry of the activation complex transitions.} The asymptotic assumption about smallness of intermediate concentration implies asymmetry between the forward and backward transitions of the activated complex (\ref{small concentrat}). Schematically, the system returns to the initial reactants much more frequently than goes forward. This asymmetry does not contradict any first principle but is a strong additional assumption. We try to relax this assumption and use the relaxation time approximation instead of the limit $\delta \to 0$ and the degenerate system.}}}
\end{figure}

What do we know about kinetics of the transitions between the reactants and the intermediates, 
\begin{equation}\label{Eq:AtoB}
\sum_i \alpha_{ i}A_i\rightleftharpoons B,
\end{equation}
in (\ref{Eq:StueckStoich}) if we do not know the kinetic law but know classical thermodynamics? 
\begin{enumerate}
\item This reaction moves the system along the straight line in the concentration space. The direction is defined by the stoichiometric vector $\gamma$ with the coordinates $\gamma_{i}=-\alpha_{i}$ for $i=1, \ldots, n$ (for the concentrations of $A_i$) and $\gamma=1$ for the concentration of $B$, all other coordinates are zeros. 
\item The reaction rate for this transition is zero when the free energy derivative in direction $\gamma$ is zero. The surface of these quasiequilibrium states can be parametrised as the function  $\varsigma=\varsigma^{\rm qe}(c,T)$, where $\varsigma$ is the concentration of $B$.
\item The reaction rate in a small vicinity of the quasiequilibrium surface can be presented in the `relaxation time' approximation:
\begin{equation}\label{Eq:LinearizReactRate}
r=\frac{1}{\tau}(\varsigma-\varsigma^{\rm qe}(c,T))+ o(|\varsigma-\varsigma^{\rm qe}(c,T)|).
\end{equation}
{\it Warning}: It is important to mention that kinetic equations with this linearised reaction rate are not linear because two reasons: (1) the non-linearity is hidden in the dependence $\varsigma^{\rm qe}(c,T)$ and (2) moreover, for a multi-stage reaction the quasiequilibrium surfaces are different and  reaction rate of each elementary stage is the result of linearisation near its own surface.
\end{enumerate}

Thermodynamic properties of a small admixture in a well-mixed homogeneous systems are similar to the perfect gas.  In the assumption that the concentration of $B$ is small, the free energy density is the sum of the free energy density $f(c,T)$ of the mixture of reactants $A_i$ and the free energy density of $B$ that has the perfect form: $RT \varsigma (\ln(\varsigma/\varsigma^*(c,T))-1)$. Here, $\varsigma^*(c,T)$ is the standard equilibrium for $B$ at the given concentrations  of reactants and temperature.

If the concentrations of intermediate $\varsigma$, $\varsigma^*$ are $\delta$-small with the partial derivative $\partial \varsigma^*/\partial c_i$ then the quasiequilibrium concentration  of  $B$ in the reaction (\ref{Eq:AtoB}) for given concentrations $c_i$ is \cite{Gorban2011}:
\begin{equation}\label{Eq:QEintermed}
\varsigma^{\rm qe}=\varsigma^*(c,T)\exp\left(\frac{\sum_i \gamma_i \mu_i}{RT}\right)+ O(\delta^2).
\end{equation}
This is the standard quasiequilibrium approximation for the concentrations of the activated complex in the TST. In what follows,  the $o$ terms are omitted.

It is convenient to represent the network of reactions (\ref{Eq:StueckStoich}) in a slightly more general form with the ensemble of equilibration reactions $\sum_i \nu_{\rho i}A_i\rightleftharpoons B_{\rho}$ and the network of monomolecular transitions between the activated complexes $B_{\rho}\to B_{\eta}$.  The reaction rates of the equilibration steps satisfy the `relaxation time' approximation (\ref{Eq:LinearizReactRate}) with the relaxation time $\tau_{\rho}$ and the transitions between activated complexes obey the first order kinetics with the reaction rate constants $\kappa_{\eta \rho}$. The notation $\kappa_{\eta \leftarrow \rho}$ can also be used to indicate the direction of the transition. The stoichiometric coefficients $\nu_{\rho i}$ can play the role of the input coefficients $\alpha_{\rho i}$ and of the output coefficients $\beta_{\rho i}$ depending of the context. The Greek indexes ($\rho$, $\eta$) in this network enumerate the formal sums $\sum_i \nu_{\rho i}A_i$. These formal sums of the reactants are called {\it complexes}. { They  are used in all formalisms of chemical kinetics. Complexes should not be confused with  activated complexes or transition states.} A systematic presentation of chemical kinetics based on the analysis of transformation graphs of complexes was given by Feinberg \cite{Feinberg2019}.  

The kinetic equations for this network can be simplified using the smallness of intermediate concentrations. After exclusion of the intermediates, we obtain the
equations:

\begin{equation}\label{result2}\boxed{
\begin{split}
&\frac{dc_i}{dt}= \sum_{\rho \eta, \, \rho\neq \eta}  \kappa_{\rho \eta}\varsigma^{\rm qss}_{\eta} (c, T) (\nu_{\rho i}-\nu_{\eta i}) ;\\
&\varsigma^{\rm qss}(c, T) =\left(1- {\rm diag}[\tau_i] K\right)^{-1}\varsigma^{\rm qe}(c, T).
\end{split}}
\end{equation}
Here, $\varsigma^{\rm qe}(c, T)$ is the vector of the quasiequilibrium concentrations of the intermediates calculated by (\ref{Eq:QEintermed}); $\varsigma^{\rm qss}(c, T)$ is the vector of the quasi steady state concentrations. The scaling parameter is omitted. $K$ is the matrix of the coefficients for the first-order kinetic equations that describe transitions between intermediates: $K_{\rho \eta}=\kappa_{\rho \eta}$ if $\rho \neq \eta$ and
$K_{\rho \rho}=- \sum_{\eta, \eta\neq \rho} \kappa_{\eta \rho }$. (The sums in the columns must be zero, which means that the total amount is preserved in the first order kinetics.) 

If all $\tau_i \to 0$ then equations (\ref{result2}) transform into the quasiequilibrium GMAL with $\varsigma^{\rm qss}=\varsigma^{\rm qe}$ given by (\ref{Eq:QEintermed}). For the general reaction scheme the semidetailed/cyclic/complex balance conditions (\ref{complexbalance}) means, at the microscopic level, just the Markov chain balance condition:
\begin{equation}\label{balanceB*}
\sum_{\eta, \rho \neq \eta}\kappa_{\rho \eta}\varsigma^*_{\eta}
=\sum_{\eta, \rho\neq \eta}\kappa_{\eta\rho}\varsigma^*_{\rho}.
\end{equation}
Here, the right hand side is the flux from $B_{\rho}$ to all other intermediates and the left hand side is the flux from all intermediates to $B_{\rho}$. The flux is evaluated at the `standard equilibrium' $\varsigma^*$. 
This identity means that $\varsigma^*$ is an equilibrium of the first order kinetics of transition network $B_{\eta}\to B_{\rho}$ with reaction rate constants $\kappa_{\rho \eta}=\kappa_{\rho \leftarrow \eta}$. 

One can apply (\ref{result2}) to real system only if the quasi steady state (qss) concentrations are non-negative: $\varsigma_i^{\rm qss}\geq 0$ for all $i$ if the concentrations $c_i$ are positive.  Let us formulate and prove the positivity theorem. Let the matrix of kinetic coefficients $K$ of the first order transitions between intermediates and the relaxation times $\tau_{\rho}>0$ of the fast quasiequilibria $\sum_i \nu_{\rho i}A_i\rightleftharpoons B_{\rho}$ be given.

\begin{theorem}[Positivity Theorem]\label{Th:Positivity}
The entanglement matrix $E=\left(1- {\rm diag}[\tau_i] K\right)^{-1}$ exists and is non-negative with strictly positive diagonal. 
\end{theorem}
\begin{proof}
The matrix $Y={\rm diag}[\tau_i^{-1}]-K$ is strictly diagonally dominant by columns with positive diagonal and non-positive non-diagonal terms. Therefore, it is non-singular and has an inverse matrix $Y^{-1}$ \cite{Varga2004}. Notice, that $Y=(1-M){\rm diag}[y_{ii}]$, where $m_{ii}=0$ and $m_{ij}=-y_{ij}/y_{jj}$ when $i\neq j$;  $m_{ij}\geq 0$ and $ 1>\sum_i m_{ij}$. The matrix $M$ is bounded in the matrix-$1$ norm (that is  the maximum absolute column sum of the matrix): $\|M\|_{1}=\max_j\sum_i |m_{ij}|<1$. Therefore, the series
$(1-M)^{-1}=1+M+M^2+\ldots$ converges. Each term in this series is non-negative and the diagonal of the sum is strictly positive. The same is true for the products of $(1-M)^{-1}$ on positive diagonals, $Y^{-1}={\rm diag}[y_{ii}^{-1}](1-M)^{-1} $ and $E=Y^{-1}{\rm diag}[1/\tau_i]$.
\end{proof}

\begin{corollary}
The  qss concentrations of intermediates are positive: $\varsigma_i^{\rm qss}>0$ for all $i$ and positive vector of concentrations $c$.
\end{corollary}
\begin{proof}
According to thermodynamic definitions (\ref{Eq:QEintermed}), the qe concentrations are positive: $\varsigma_i^{\rm qe}>0$  for all $i$. Theorem~\ref{Th:Positivity} guarantees that the vector of the qss concentrations is also positive, $\varsigma^{\rm qss}=E\varsigma^{\rm qe}>0$.
\end{proof}

Let us demonstrate the formula (\ref{result2}) on the simplest toy example. Consider the system with three components, $A_{1-3}$  with concentrations $c_{1-3}$, and  two reversible elementary `brutto'  reactions, $A_1\rightleftharpoons A_2 \rightleftharpoons A_3$. Introduce three intermediates, $B_{1-3}$  with concentrations $\varsigma_{1-3}$. For (\ref{result2}) we need four reaction rate constants for the transitions between intermediates, $B_{\rho} \to B_{\eta}$, three standard equilibria $\varsigma_{\rho}^*$ that satisfy the balance conditions (\ref{balanceB*}):
$$\kappa_{12}\varsigma_{2}^*=\kappa_{21}\varsigma_{1}^*; \; \kappa_{23}\varsigma_{3}^*=\kappa_{32}\varsigma_{2}^*,$$
and the free energy density of reactants $f(c,T)$. For the example, we select the perfect free energy $f=RT\sum_i c_i (\ln(c_i/c_i^*)-1)$. Than the qe concentrations of reactants are $\varsigma^{\rm qe}_i=\varsigma^*_i c_i/c_i^*$.

Let us notice, that for the linearly independent reactions the balance conditions (\ref{balanceB*}) turn into the detailed balance conditions. 
The matrices $K$ and $E^{-1}=1- {\rm diag}[\tau_i] K$ are:
\begin{equation*}
K=\left(\begin{array}{lll}
-\kappa_{21} & \kappa_{12} &                 0 \\
\kappa_{21}  & -\kappa_{12}-\kappa_{32} &\kappa_{23} \\
0            & \kappa_{32}            & -\kappa_{23}   
\end{array} \right); 
\end{equation*} 

\begin{equation*}
E^{-1}=\left(\begin{array}{lll}
1+ \tau_1 \kappa_{21} & -\tau_1 \kappa_{12} &                 0 \\
-\tau_2 \kappa_{21}  & 1+ \tau_2(\kappa_{12}+\kappa_{32}) & -\tau_2 \kappa_{23} \\
0            & -\tau_3 \kappa_{32}            & 1+\tau_3 \kappa_{23}   
\end{array} \right).
\end{equation*}

According to the Positivity Theorem, $E$ is a non-negative matrix with strongly positive diagonal. The kinetic model (\ref{result2}) for the system $A_1 \rightleftharpoons A_2 \rightleftharpoons A_3$ is

\begin{equation}\label{Eq:SimpleKin}
\begin{split}
& \frac{dc_1}{dt}=-\kappa_{21}\varsigma^{\rm qss}_1+\kappa_{12}\varsigma^{\rm qss}_2; \\
& \frac{dc_2}{dt}=\kappa_{21}\varsigma^{\rm qss}_1-(\kappa_{12}+\kappa_{32})\varsigma^{\rm qss}_2 +\kappa_{23}\varsigma_3^{\rm qss}; \\
& \frac{dc_3}{dt}=\kappa_{32}\varsigma^{\rm qss}_2-\kappa_{23}\varsigma^{\rm qss}_3 \\
&\varsigma^{\rm qss}(c, T) =E\, \varsigma^{\rm qe}(c, T).
\end{split}
\end{equation} 

If the life time of the intermediates is small and  $\tau_i \to 0$, then $\varsigma^{\rm qss}_i= \varsigma^{\rm qe}$ and, for the perfect systems, 
$\varsigma^{\rm qss}_i=\varsigma^*_i c_i/c_i^*$. In this case, the system (\ref{Eq:SimpleKin}) is just a usual first order kinetic equation (a continuous time Markov chain) with two reversible transitions. 

For the non-negligible $\tau_i$, the situation seems to be more sophisticated and the reaction rates will be entangled by the matrix $E$. To demonstrate the result of  entanglement, let us calculate $E$ for a very simple symmetric case $\tau_i=1$ and $
\kappa_{12}=\kappa_{21}=\kappa_{23}=\kappa_{32}=1$. In this case, the balance condition (\ref{balanceB*}) gives: $\varsigma^*_1=\varsigma^*_2=\varsigma^*_3\,(=\varsigma^*>0)$. Only one parameter of these three  $\varsigma^*_i$ is free. For the qe concentrations we get $\varsigma^{\rm qe}_i=\varsigma^* c_i/c_i^*$.

\begin{equation*}
\begin{split}
&E^{-1}=\left(\begin{array}{lll}
2 & -1 &    0 \\
-1  & 3 & -1 \\
0            & -1            & 2   
\end{array} \right);
E=\frac{1}{8}\left(\begin{array}{lll}
5 & 2 &   1 \\
2  & 4 & 2 \\
1 & 2  & 5   
\end{array} \right);\\
&\varsigma_1^{\rm qss}=\frac{\varsigma^*}{8}\left(5\frac{c_1}{c_1^*}+2\frac{c_2}{c_2^*}+\frac{c_3}{c_3^*}\right),\\ 
&\varsigma_2^{\rm qss}=\frac{\varsigma^*}{8}\left(2\frac{c_1}{c_1^*}+4\frac{c_2}{c_2^*}+2\frac{c_3}{c_3^*}\right),\\ 
&\varsigma_3^{\rm qss}=\frac{\varsigma^*}{8}\left(\frac{c_1}{c_1^*}+2\frac{c_2}{c_2^*}+5\frac{c_3}{c_3^*}\right). 
\end{split}
\end{equation*}

\begin{equation}\label{Eq:SimplestKin}
\begin{split}
& \frac{dc_1}{dt}=\frac{\varsigma^*}{8}\left(-3\frac{c_1}{c_1^*}+2\frac{c_2}{c_2^*}+\frac{c_3}{c_3^*}\right); \\
& \frac{dc_2}{dt}=\frac{\varsigma^*}{8}\left(2\frac{c_1}{c_1^*}-4\frac{c_2}{c_2^*}-2\frac{c_3}{c_3^*}\right); \\
& \frac{dc_3}{dt}=\frac{\varsigma^*}{8}\left(\frac{c_1}{c_1^*}+2\frac{c_2}{c_2^*}-3\frac{c_3}{c_3^*}\right).
\end{split}
\end{equation} 

Compare (\ref{Eq:SimplestKin}) to the system with the same $\kappa_{\eta \rho}$ and $\varsigma^*_i$ but all $\tau_i=0$ (\ref{Eq:SimplestKinQE}):  

\begin{equation}\label{Eq:SimplestKinQE}
\begin{split}
& \frac{dc_1}{dt}=\varsigma^*\left(-\frac{c_1}{c_1^*}+\frac{c_2}{c_2^*}\right); \\
& \frac{dc_2}{dt}=\varsigma^*\left(\frac{c_1}{c_1^*}-2\frac{c_2}{c_2^*}+\frac{c_3}{c_3^*}\right); \\
& \frac{dc_3}{dt}=\varsigma^*\left(\frac{c_2}{c_2^*}-\frac{c_3}{c_3^*}\right).
\end{split}
\end{equation} 
We can see that if the equilibrium between the reactant complexes (here, these complexes are just the reactants in themselves) and the corresponding intermediates are not infinitely fast comparing to transition between complexes then the new reaction appears, $A_3 \rightleftharpoons A_1$, and all the observable reaction rate constants should be redefined. The kinetic equations (\ref{Eq:SimplestKin}) remain the classical chemical kinetic equations but with additional reactions and new reaction rate constants. The question arises: why the non-standard reasoning  returns the MAL (or GMAL) kinetic equations? In the next section we answer this question for general reaction networks.

\section{Demystification \label{Sec:Demyst}}

The entangled GMAL equations (\ref{result2}) are produces by the QSS approximation of the complex reaction network that includes:
\begin{itemize} 
\item  {\it Unknown} kinetics of transitions between the complexes of reactants and the corresponding intermediates $\sum_i \nu_{\rho i} A_i \rightleftharpoons B_{\rho}$ linearised near their thermodynamic equilibria (quasiequilibria - equilibria along the reaction stoichiometric vectors);
\item A network of transitions between the intermediates described by linear kinetic equations (first order kinetics or Markov chain).
\end{itemize} 
Under the assumption that concentrations of $B_{\rho}$ ($\varsigma_{\rho}$) are sufficiently small, the linearised rate of transition is found in the form (\ref{Eq:LinearizReactRate})
$ r_{\rho}= \frac{1}{\tau_{\rho}}(\varsigma_{\rho}-\varsigma_{\rho}^{\rm qe})$, 
where $\varsigma_{\rho}^{\rm qe}$ are given by the thermodynamic formulas  (\ref{Eq:QEintermed}). We should notice now that the linearised reaction rate is exactly the GMAL reaction rate for the same reaction with 
\begin{equation}\label{Eq:Demyst}
r^-_{\rho}=\frac{1}{\tau_{\rho}}\varsigma_{\rho}, \;\; r^+_{\rho}=\frac{1}{\tau_{\rho}}\varsigma_{\rho}^*\exp\sum_i(\nu_{\rho i} \frac{\mu_i}{RT}).
\end{equation} 

Thus, after linearisation the {\it unknown} kinetics of transitions between complexes of reactant and intermediates near quasiequilibria and neglecting some high order terms (in the concentrations of intermediates), we obtain the GMAL system shown in Fig.~\ref{Fig:NetIntermed} with reaction rates (\ref{Eq:Demyst}). This reaction network was further simplified by the QSS approximation, where the small parameter is in the concentrations of the intermediates. 

It is not a miracle that starting with a GMAL network and applying the QSS approximation we again get a (reduced) GMAL network. This network (Fig.~\ref{Fig:NetIntermed}) is the main intermediate construction for deriving macroscopic chemical kinetic equations. It resembles the networks studied in the kinetics of catalytic reactions \cite{Marin2019}, and various networks of intermediates can give a rich class of the brutto kinetics equations in the QSS approximation.

\begin{figure}
\centering{
\includegraphics[width=\columnwidth]{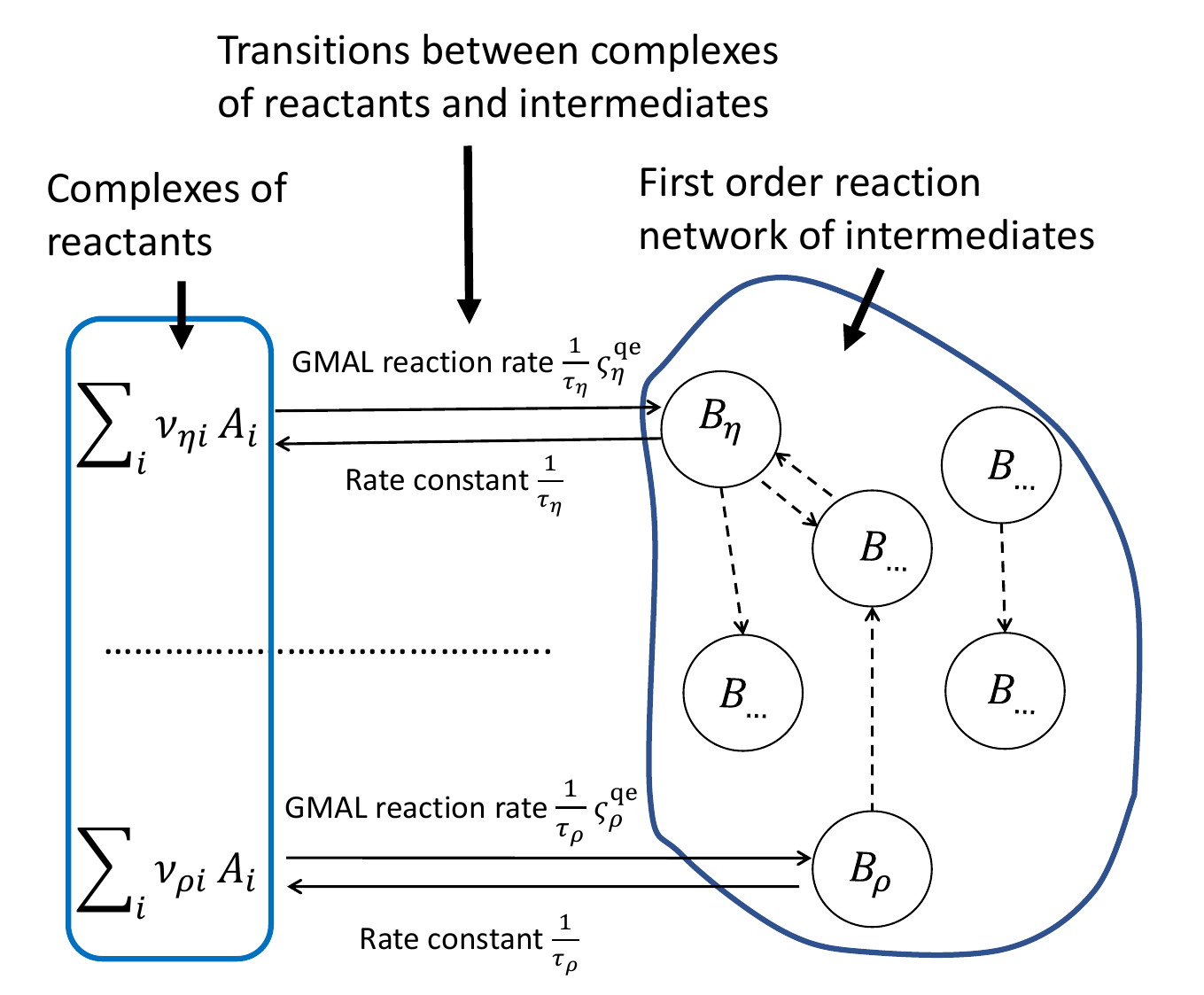}
 \caption{\label{Fig:NetIntermed}{\it TST auxiliary reaction network.} {  A network of GMAL reaction that corresponds to the reaction network with unknown kinetics of the transitions $\sum_i \nu_{\rho i} A_i \rightleftharpoons B_{\rho}$ linearised near their thermodynamic equilibria (quasiequilibria - equilibria along the stoichiometric vectors) and simplified using the smallness of intermediate concentrations. The first order reaction network between intermediates is shown on the right, reactant complexes on the left, and GMAL transitions between complexes and intermediates are shown in the middle.}}}
\end{figure}

Despite some similarity  between the TST auxiliary networks presented in Fig.~\ref{Fig:NetIntermed} and the networks of catalytic 
reactions \cite{Marin2019,YBGE1991} there is an important difference. In catalytic reactions, the intermediates, which include the catalysts and their compounds, participate in the input complexes. There are, for example, reactions like $A+Z\to AZ$, where $A$ is the reactant and $Z$ is the catalyst. In the TST auxiliary network (Fig.~\ref{Fig:NetIntermed}) the input complexes $\sum_i \nu_{\rho i} A_i$ do not include the intermediates. Therefore, the quasi-stationary asymptotics looks different for them.

The auxiliary reaction network presented in Fig.~\ref{Fig:NetIntermed} gives a convenient description of a complex reaction with fast intermediates that are not very far from the quasiequilibrium with the reactants and have small concentrations. The whole chain of simplifications from the reaction mechanism (\ref{Eq:StueckStoich}) with unknown kinetic law to the entangled GMAL equations (\ref{result2}) is presented in Fig.~\ref{Fig:FlowchartDemyst}.

\begin{figure}
\centering{
\includegraphics[width=\columnwidth]{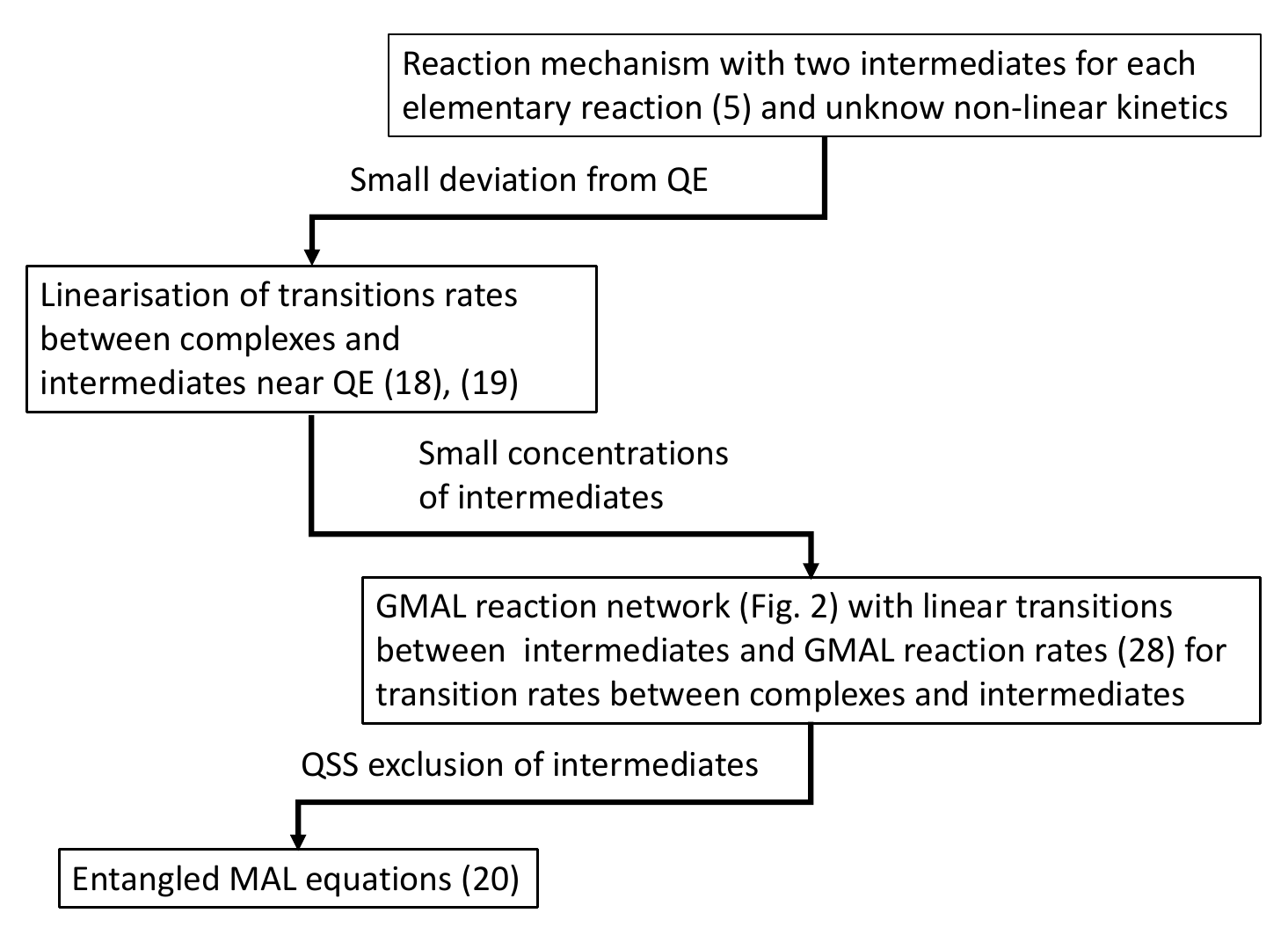}
 \caption{\label{Fig:FlowchartDemyst}{ {\it Asymptotic assumptions and stair of simplifications.} A stair  of simplifications is presented from the reaction mechanism (\ref{Eq:StueckStoich}) to the entangled GMAL equations (\ref{result2}). Arrows are labeled with asymptotic assumptions and the simplification results are described in rectangles.}}}
\end{figure}

\section{Conclusion}

{

The central part of TST is the quantum (or quasiclassical) theory of the activated complex and its transition. To embed these results in the broader context of chemical kinetics, additional dynamic and thermodynamic assumptions are needed. The standard approach uses an explicit quasiequilibrium assumption implicitly supplemented  by the assumption of small concentrations of activated complexes. In this paper, we analysed some aspects of the thermodynamic and kinetic background of TST, but do not touched on the quantum theory of the activated complex.

First, we demonstrated on a textbook example of the TST model how the  assumption of quasiequilibrium works without of low intermediate concentration. We used the assumption of quasiequilibrium to exclude the concentration of the activated complex from the kinetic equations. If the concentration of the activated complex is not low, then it can serve as a reservoir of the reactants. In this case, the resulting reaction rate differs from MAL and becomes more complex (\ref{Eq:quadraticQE}), (\ref{EqTNTkin}).

It was clearly demonstrated in the previous works \cite{Gorban2021, Perez-Benito2017} that the standard TST kinetics with the quasiequilibrium assumption and one intermediate, for example, $A+B \rightleftharpoons [A-B] \to P$, did not work for reversible reactions. Indeed, due to detailed balance  the reverse reaction should go through the same activated complex, $A+B \rightleftharpoons [A-B] \rightleftharpoons P$, and two quasiequilibria $A+B \rightleftharpoons [A-B]$ and $[A-B] \rightleftharpoons P$ together imply equilibrium and zero reaction rate. The solution to this paradox for gas kinetics is known since 1952 when Stueckelberg proposed {\it two activated complexes} model for collisions in the Boltzmann equation \cite{Stueckelberg1952}: $$A_1+A_2 \rightleftharpoons B_1 \rightleftharpoons B_2 \rightleftharpoons A_3+A_4.$$ 
The connection of this seminal work with chemical kinetics was revealed only in 2011 \cite{Gorban2011}.

Analysis of the asymptotic assumptions for simple kinetics demonstrates a significant asymmetry in the reaction rate constants for the intermediates even for the two-intermediate  models: in the transitions $A_1+A_2  \rightleftharpoons B \rightleftharpoons \ldots $ the reverse reaction
$ A_1+A_2 \leftarrow B$ should be faster than $B\to \ldots$ (Fig.~\ref{Fig:Asymmetry}). Without that asymmetry, the quasiequilibrium asymptotic does not work. Such asymmetry does not contradict any first principles but seems to be a very strong arbitrary assumption. To relax this assumption, quasiequilibrium relations are substituted by the relaxation time approximation near quasiequilibrium (\ref{Eq:LinearizReactRate}). These relaxation time models together with the smallness of intermediates led to the entangled mass action law equations (\ref{result2}) \cite{Gorban2021}.

One important issue was missed in \cite{Gorban2021}: positivity of the intermediate concentrations was not proven. It is well-known that manipulations with small parameters and asymptotics can result in violation of basic principles (for example, multiple equilibria and bifurcation may appear in the QSS asymptotic  even in kinetics of closed systems). We proved that the entangled mass action law equations (\ref{result2}) preserve positivity of the intermediate concentrations (Theorem~\ref{Th:Positivity}). 

Another novelty is in the backgrounds of the  entangled mass action law equations for reaction networks. We demonstrated that the relaxation time approximation for the reactions
$$\sum_i \nu_{\rho i} A_i \rightleftharpoons B_{\rho}$$ near quasiequilibrium (\ref{Eq:LinearizReactRate}) turns into the well-known GMAL equations for low intermediate concentrations.
We assume that these elementary reactions satisfy the thermodynamic restrictions and material balance conditions. The linear approximation of this a priory unknown kinetics near its partial equilibrium was used in \cite{Gorban2021} for derivation of (\ref{result2}). Surprisingly, under the same conditions, the kinetic law  (\ref{Eq:LinearizReactRate}) produced from an arbitrary kinetics coincided  with GMAL. This result is presented in Fig.~\ref{Fig:NetIntermed}. This figure on the left shows a list of reactant complexes. These are formal sums taken from the stoichiometric equations of elementary reactions \cite{Feinberg2019}. On the right, there is  a network of first order transitions between intermediates. The vertices of this network correspond to the intermediates $B_{\rho}$, $B_{\eta}, \ldots$. In the middle, the transitions between complexes and intermediates are showed.  Each complex is connected with an intermediate by a reversible reaction. The reaction rate of these reactions satisfy GMAL. 

The main novelty of this scheme is that the assumption of small concentrations of intermediates and the linearisation of arbitrary and unknown kinetics of transitions between complexes and intermediates led to a network of reactions (Fig.~\ref{Fig:NetIntermed}) with well-known GMAL kinetics between them. We aim to propose this asymptotically well-founded scheme as a standard for embedding TST results in the kinetics of complex reactions.

}

\section*{Declarations}

\subsection*{Author contribution statement:}

A. N. Kirdin; S. V. Stasenko: Conceived and designed the analysis; Analyzed equations and proved the results; Wrote the paper.

\subsection*{Data availability statement:}

    No data were used for the research described in the article.

\subsection*{Declaration of interest's statement:}   

   The authors declare no competing interests.
   
\section*{Funding}
The work was supported in part in 2021 by  the Ministry of Science and Higher Education of Russian Federation (Project No. 075-15-2021-634).

\end{document}